\documentclass{article}

\PassOptionsToPackage{hidelinks}{hyperref}

\usepackage[preprint]{neurips_2026}

\usepackage[utf8]{inputenc} 
\usepackage[T1]{fontenc}    
\usepackage{hyperref}       
\usepackage{url}            
\usepackage{booktabs}       
\usepackage{amsfonts}       
\usepackage{nicefrac}       
\usepackage{microtype,bbm}      
\usepackage{xspace}         
\usepackage[symbol]{footmisc}
\usepackage[table,xcdraw]{xcolor}
\usepackage[english]{babel}

\usepackage{subcaption}
\usepackage{algorithm}
\usepackage[noend]{algpseudocode}

\usepackage{amsmath}
\allowdisplaybreaks
\usepackage{amssymb}
\usepackage{mathtools}
\usepackage{amsthm}
\usepackage{thm-restate}
\usepackage[autostyle=true]{csquotes}
\usepackage{enumitem}
\usepackage{cleveref}

\newcommand{\xhdr}[1]{\vspace{1mm} \noindent{\bf #1}}

\newcommand{\cA}{\mathcal{A}}
\newcommand{\cD}{\mathcal{D}}
\newcommand{\cG}{\mathcal{G}}
\newcommand{\cH}{\mathcal{H}}
\newcommand{\cL}{\mathcal{L}}
\newcommand{\cM}{\mathcal M}
\newcommand{\cN}{\mathcal N}

\newcommand{\cS}{\mathcal{S}}
\newcommand{\cT}{\mathcal T}

\newcommand{\vs}{v_s}
\newcommand{\vb}{v_b}
\newcommand{\bs}{b_s}
\newcommand{\bb}{b_b}

\newcommand{\ind}[1]{\mathbbm{1}_{\{{#1\}}}}
\newcommand{\E}[1]{\mathbb{E}\left[#1\right]}

\renewcommand{\P}[1]{\mathbb{P}\left(#1\right)}
\newcommand{\Psub}[2]{\mathbb{P}_{#1}\left(#2\right)}

\newcommand{\eps}{\varepsilon}

\newcommand{\hedge}{\textsc{Hedge}\xspace}
\newcommand{\hier}{\textsc{Hier-Mech}\xspace}
\newcommand{\hhedge}{\textsc{HierHedge}\xspace}

\newcommand{\cMbl}{\cM_{\textsc{BT}}}
\newcommand{\cAja}{\cA_{\textsc{JA}}}
\newcommand{\cAbl}{\cA_{\textsc{BT}}}
\newcommand{\blM}{M_{\textsc{BT}}}
\newcommand{\jaM}{M_{\textsc{JA}}}
\newcommand{\blA}{A_{\textsc{BT}}}
\newcommand{\jaA}{A_{\textsc{JA}}}

\newcommand{\prof}{\textnormal{\textsf{profit}}}
\newcommand{\rev}{\textnormal{\textsf{rev}}}
\newcommand{\app}[2]{\phi_{#1}\left(#2\right)}
\newcommand{\best}{M^\star}
\newcommand{\bestapp}[1]{\phi^\star_{#1}}
\newcommand{\util}[3]{U_{#3}(#1, #2)}

\newcommand{\target}[2]{\theta_{#1}({#2})}
\newcommand{\fat}[2]{\text{\normalfont fat}_{#1}(#2)}

\newtheorem{definition}{Definition}
\newtheorem{proposition}{Proposition}
\newtheorem{claim}{Claim}
\newtheorem{lemma}{Lemma}
\newtheorem{theorem}{Theorem}

\title{Profit Maximization in Bilateral Trade against a Smooth Adversary}

\author{%
  Simone~{Di~Gregorio}\\
  Department of Computer,\\ Control and Management Engineering\\
  Sapienza University of Rome\\
  \texttt{simone.digregorio@uniroma1.it} 
    \And
    Paul~{D\"utting}\\
  Google Research\\
  \texttt{duetting@google.com} 
  \And
    Federico~{Fusco}\\
    Department of Computer,\\ Control and Management Engineering\\
  Sapienza University of Rome\\
    \texttt{federico.fusco@uniroma1.it} \And
    Chris~{Schwiegelshohn}\\
  Department of Computer Science\\ Aarhus University\\
  \texttt{cschwiegelshohn@gmail.com}
}

\begin{document}

\maketitle

\begin{abstract}
    Bilateral trade models the task of intermediating between two strategic agents, a seller and a buyer, who wish to trade a good.  
    We study this problem from the perspective of a profit-maximizing broker within an online learning framework, where the agents' valuations are generated by a smooth adversary. 

    We devise a learning algorithm that guarantees a $\tilde{O}(\sqrt{T})$ regret bound, which is tight in the time horizon $T$ up to poly-logarithmic factors. This matches the minimax rate for the stochastic i.i.d.~case, and is also well separated from the adversarial setting, where sublinear-regret is unattainable. By extending the strong regret guarantees from the i.i.d.~case to the smooth adversary, we significantly broaden the scope of settings where such fast rate 
    is achievable, while closing an important gap in the regret landscape of this fundamental economic problem. 
    
    To overcome the 
    challenges posed by this adversary, we leverage a continuity property of smooth instances and combines this with a hierarchical net-construction of the broker's action space, which is analyzed via algorithmic chaining.
    We showcase the applicability of these techniques by deriving a similarly tight $\tilde{O}(\sqrt{T})$ regret bound for a related mechanism design model: the joint ads problem.
\end{abstract}


\section{Introduction}

    In bilateral trade, a broker intermediates between two agents---a seller and a buyer---who are interested in trading a good \citep{Vickrey61,MyersonS83}. Both agents are characterized by a private valuation for the item, and their goal is to maximize their own utility. The broker, on the other hand, aims at maximizing some economic metric, 
    while ensuring that (i) agents report truthfully their valuations (\emph{dominant-strategy incentive compatibility}) and (ii) the utility for participating in the mechanism of each agent is non-negative (\emph{individual rationality}). 

    In this paper, we focus on \emph{profit-maximization}, a canonical objective in economics which naturally arises in the applications (e.g., all e-commerce platforms that connect sellers to buyers). In particular, we investigate profit maximization in bilateral trade through the lens of online learning, with the goal of designing an algorithm that minimizes the regret with respect to the best mechanism in hindsight. For this problem, \citet{DiGregorioDFS25} provides two surprising results: while the best mechanism can be learned \emph{very} quickly in the stochastic setting where the agents valuations are drawn i.i.d. from an unknown but fixed distribution, \emph{no algorithm} can achieve sublinear regret in the standard adversarial setting. This leaves open the natural question of whether it is possible to perform learning \emph{beyond} the stationary i.i.d. setting. Following a consolidated approach rooted in \emph{beyond-worst-case-analysis} \citep{Roughgarden20}, we investigate the learnability of the problem against the $\sigma$-smooth adversary. Such an adversary can choose different distributions over valuations at different time steps, but is constrained to distributions that are ``not too peaked'' (see \Cref{def:smoothness}).

    A characterizing feature of our learning task is the complexity of the broker action space. A mechanism for bilateral trade is specified by two functions---an allocation and a payment rule---mapping bids to allocation and payments. Standard mechanism design arguments (see Appendix~\ref{app:bilateral_trade}) allow for a simple geometric characterization of the generic $M \in \cM$. Its allocation rule corresponds to a subdivision of the bids space into two regions: one in which the trade happens and one where this is not the case. This partition has to be closed in the ``up-left'' direction, so that increasing the buyer's bid or decreasing the seller's does not decrease the likelihood of observing a trade. The payments are then given by suitable projection of the valuation pair on the boundary of the allocation region, à la \citet{Myerson81}. While this characterization collapses the learning problem to tracking the best allocation region, the mechanism space is still non-parametric! This separates us from the majority of the online learning literature on mechanism design, where the learning objective is intrinsically parametric, e.g., the best price(s) \citep{KleinbergL03,Cesa-BianchiCCF24mor,Cesa-BianchiCCF24jmlr}, the best reserve price(s) \citep{Cesa-BianchiGM15,RoughgardenW19}, or the best bid \citep{FengPS18,Cesa-BianchiCCF24}.

    \subsection{Our Result}

        Our first contribution is the design of the learning algorithm \hier, which enjoys a regret rate of $O(\nicefrac{1}\sigma \cdot \sqrt{T} \log T)$ for profit maximization in bilateral trade against a $\sigma$-smooth adversary (\Cref{thm:hierHedge}). This is, up to lower order terms, the same bound achievable in the i.i.d. setting, while no sublinear algorithm can exist in the (non-smooth) adversarial setting \citep{DiGregorioDFS25}.
        
        We then formally prove a reduction from the joint ads problem to profit-maximization in bilateral trade (\Cref{thm:reduction}). Joint ads is a mechanism design setting where two agents \emph{collaborate} to purchase a non-excludable good, 
        and has been studied in the online learning setting in \citet{AggarwalBDF24}. The reduction implies a lower bound of $\Omega(\sqrt{T})$  on the regret for bilateral trade (\Cref{cor:lower_bilateral}), matching the performance of \hier. Therefore, we characterize the time-horizon dependence of the minimax regret rate against the $\sigma$-smooth adversary, up to poly-log terms. This settles an open question in \citet{DiGregorioDFS25}. Finally, the reduction allows us to adapt \hier to the joint ads setting, maintaining the same (tight) regret rate. This improves on the previous regret upper bound of $O(\nicefrac{1}{\sigma} \cdot T^{\nicefrac 23})$, settling an open question in \citet{AggarwalBDF24} (see \Cref{cor:upper_joint}).

        
        In combination, our results thus show that, from an information-theoretic perspective, the $\sigma$-smooth adversary is ``as easy as'' the i.i.d.~setting. Interestingly, the $\sqrt{T}$ rate algorithms for these two data-generation models are obtained with very different techniques and arguments 
        (see \Cref{subsec:challenges} for more details). 
        The strong statistical learning guarantees for the smooth adversarial setting reinforce the effectiveness of smoothness as a meaningful way of overcoming adversarial impossibility results in economic problems \citep{Cesa-BianchiCCF24mor,Cesa-BianchiCCF24jmlr,Cesa-BianchiCCF24,DurvasulaHZ23}.
        A remaining discrepancy is that the running time of \hier scales as $\exp(\sqrt{T})$, while the i.i.d.~setting admits a poly-time algorithm. Closing this computational gap and tightening the regret dependency on the smoothness parameter $\sigma$ remain the primary open questions in this space.


    \subsection{Technical Challenges}
    \label{subsec:challenges}

    The main challenge of profit-maximization in bilateral trade is the non-parametric nature of the action space, and the special ``projection-like'' nature of the payments, which induce the profit function.


        \xhdr{The Adversary is Non-Stationary.} In the stochastic i.i.d. setting, \citet{DiGregorioDFS25} address the learning problem by proposing a modified version of follow-the-leader, arguing via \emph{probabilistic chaining} \citep{Talagrand14} that (i) a data and distribution dependent \emph{finite} family $\cM'$ of mechanisms has its profit estimated very quickly and (ii) all mechanisms outside $\cM'$ can be ruled out quickly as candidate optimum. Such approach, however, is doomed to fail in any non-stationary setting, like the one we are studying. Indeed, for any mechanism $M$, past profits can be completely uncorrelated with future ones, so that one can never ``rule out'' mechanisms which are suboptimal up to that point.

        \xhdr{Exploiting Lipschitzness.} Smoothness makes the expected-profit function Lipschitz in the choice of the mechanism (for a proper distance). Following the standard approach in the Lipschitz bandits literature \citep{KleinbergSU19}, we can then restrict our attention to a suitable metric cover and run some off-the-shelf algorithm on it. A natural candidate is the family of $\eps$-grid mechanisms $\cM_{\eps}$, whose allocation region is made up of tiles of the uniform $\eps$-grid of the unit square $[0,1]^2$ (see \Cref{fig:approximations}). Any mechanism $M \in \cM$ can be approximated by some $M_{\eps} \in \cM_{\eps}$, so that the overall expected profit of $M$ is at least that of $M_{\eps}$, up to an additive $\nicefrac{\eps}{\sigma}$ term (as in e.g., Lemma 5.1 of \citet{AggarwalBDF24}). Unfortunately, applying directly a no-regret algorithm for adversarial prediction with experts on $\cM_{\eps}$ (e.g., \hedge) yields a suboptimal rate (as $\log |\cM_{\eps}| \approx \nicefrac{1}{\eps}$ , see  \Cref{lem:card_bound}):
        \[
            R_T \lesssim \!\! \underbrace{\nicefrac{\eps}{\sigma} \cdot T}_{{\text{Discretization}}} + \!\!\underbrace{\sqrt{T \log |\cM_{\eps}|}}_{\hedge \text{ guarantee in }\cM_\eps} \lesssim \frac{\eps}{\sigma} T + \sqrt{\frac{T}{\eps}} \approx \frac{1}{\sigma} T^{\nicefrac 23}. \tag{For $\eps \approx T^{-\nicefrac 13}$}
        \]

        \xhdr{Algorithmic Chaining.} To get $\sqrt{T}$, we exploit the specific structure of the mechanism space $\cM$. In particular, we adopt \emph{algorithmic chaining} \citep{Cesa-BianchiGGG17} to mimic \emph{in an adversarial setting} what probabilistic chaining does in the stochastic model. At a high level, the idea is to maintain a hierarchical decomposition of the mechanism space composed by suitable \emph{nets} at different precisions, and to construct an ``hedge-like'' distribution via a random walk on such structure. This allows to get a ``Dudley-type'' bound \citep[e.g., Chapter 8 in][]{VershyninHDP2}, improving on the standard $\sqrt{T \log N}$ expert bound (with $N$ being the number of experts).
         While our approach is inspired by the \hhedge algorithm \citep[Appendix H of][]{Cesa-BianchiGGG17}, there are several important differences and new challenges that we need to overcome. Indeed, algorithmic chaining is an algorithmic \emph{template}, which still needs to be instantiated with the ``right'' choices of losses/rewards and hierarchical decomposition of the action space, and such choice is naturally problem-dependent. In particular, while in \citet{Cesa-BianchiGGG17} the authors require \emph{deterministic} $L^\infty$ net guarantees on the hypothesis class, we know that such nets cannot exist for our problem (see Appendix~\ref{app:fat-shattering} for details). We overcome this issue by providing probabilistic net guarantees, obtained via a novel $L^1$ bound in expectation (\Cref{lem:l1_net}) for $\eps$-grid mechanisms, which relies on smoothness and bilateral trade-specific features. To the best of our knowledge, we are the first to combine algorithmic chaining with smoothness.
        


        
        \xhdr{Joint Ads Reduction.} While \emph{some} {relation}
        between bilateral trade and  
        joint ads is already hinted at
        in previous work \citep{DiGregorioDFS25}, 
        we derive a formal reduction from the latter to the former. 
        This allows us to apply 
        our algorithm to joint ads (implying
        a tight $\tilde{O}(\nicefrac{1}{\sigma}\cdot\sqrt{T})$ regret rate); moreover, 
        we
        extend 
        the $\Omega(\sqrt{T})$ 
        joint ads lower bound of \citet{AggarwalBDF24} to bilateral trade. 

    \subsection{Additional Related Work}

    \xhdr{Repeated Bilateral Trade.} The bilateral trade problem has been extensively investigated in the regret minimization setting when the underlying goal is to maximize welfare (or, equivalently, gain from trade). A list of work \citep{Cesa-BianchiCCF24mor,Cesa-BianchiCCF24jmlr,AzarFF24,BernasconiCCF24,LunghiC026} has characterized the minimax regret rates achievable under various feedback models, data generation models and budget balance constraints (when the goal is welfare, it makes sense to impose the additional constraint that the broker should \emph{not lose money}). Note, in all this literature, the regret benchmark is the best fixed price, as opposed to our rich class of mechanism. Conversely, in profit maximization we need access to the agents valuations to compute payments, while for fixed-price mechanisms it is natural to study less informative feedback structures. We observe that smoothness has already played an important role in this works, as an important tool to avoid pathological ``needle-in-an-haystack'' lower bounds while remaining a meaningful and challenging problem \citep{Cesa-BianchiCCF24mor,Cesa-BianchiCCF24jmlr}.

    \xhdr{Smoothed Online Learning.} In the last ten years, a growing body of work --- starting from \citet{RakhlinST11} --- has been focusing on the smoothness assumption to overcome the common impossibility results from adversarial online learning. Indeed, while general online learnability is characterized via very strong sequential complexity measures like the Littlestone dimension~\citep{shai2009,sequential2015,tightsequential2021}, the smoothness assumption for the adaptive adversary collapses the statistical hardness of online classification and regression back to the well-known VC/fat shattering dimension~\citep{HaghtalabRS20,HaghtalabRS24,BlockDGR22}. Notably, this is not true for multiclass classification with unbounded label space, when facing a slightly stronger smooth adversary~\citep{raman2024}. More recently, \citet{moise2025} and \citet{ermrakhlin2024} enriched and generalized this information-theoretical viewpoint by investigating the smooth online learning setting when the underlying measure is not assumed to be the Lebesgue measure and it is not assumed to be known to the learner, again obtaining VC-dependent regret rates. 
    
    Since our hypothesis class is not arbitrary and our learning objective (i.e. the optimal mechanism) is dependent on the valuation distribution, our setting is not captured by the lower bounds proven in this line of work; this is part of the reason why learning is for us achievable even if the fat shattering dimension is unbounded. Moreover, since the fat shattering dimension is unbounded (see Appendix \ref{app:fat-shattering}), the positive results and associated techniques are mostly inapplicable to our case. However, we do encounter the same information-theoretic separation between smooth and general adversaries, in the sense of smoothness being the enabler of a non-trivial minimax rate.
    
    From a computational perspective, a different but related line of work considers the problem of attaining optimal statistical rates under the smooth assumption while minimizing the time complexity of the algorithm used to attain them, either explicitly --- for certain model specifications~\citep{block2022} --- or implicitly in terms of calls to an empirical risk minimization oracle~\citep{BlockDGR22,haghtalabefficient2022}.
    
    \xhdr{Mechanism Design  and Online Learning.} Modern computational economics and learning theory have a strong intersection, mostly starting from \citet{KleinbergL03}, with a variety of economic problems which have been studied through this lens. For example, auctions~\citep{BlumKRW03,RoughgardenW19,kumarfirstprice,vitercik2022}, market making~\citep{cesabianchimarketmaking2025}, contracts~\citep{icmlpaulcontracts2023}, brokerage~\citep{brokerageroberto2025} and Bayesian persuasion~\citep{persuasion1,persuasion2} were all investigated from the online learning perspective.


\section{Model and Preliminaries}
\label{sec:model}

    We study the online task of learning the profit maximizing mechanism for bilateral trade, when the agents valuations are generated by a smooth adversary. We focus on the family $\cM$ of all the mechanisms which enforce dominant strategy incentive compatibility (DSIC) and ex-post individual rationality (IR). These two standard properties ensure that (i) the seller and the buyer declare their private utility truthfully and (ii) they do not suffer negative utility in participating to the mechanism. We refer to Appendix~\ref{app:bilateral_trade} for 
    more details on the game-theoretic aspects of the problem. 
    
    The seller valuation $\vs$ and the buyer one $\vb$ belong to $[0,1]$, so a mechanism $M$ is characterized by an allocation region $A \subseteq[0,1]^2$, and pricing rules $p ,q:[0,1]^2 \to [0,1]$. 
    The trade happens if and only if the agents' valuations $(\vs,\vb)$ belong to  $A$, and payments are made according to $p$ and $q$. 
    That is, the seller receives $p(\vs,\vb)$, 
    while the buyer pays $q(\vs,\vb)$.
    Without loss of generality, we require that $p(\vs,\vb) = q(\vs,\vb) = 0$ whenever $(\vs,\vb) \not\in A$, i.e., there is no trade. The profit earned by the broker running mechanism $M$ 
    is the difference between the payment the mechanism collects from the buyer and the payment that it makes to the seller:
    \(
        \prof(M,(\vs,\vb)) = q(\vs,\vb) - p(\vs,\vb). 
    \)

    \xhdr{Characterization of DSIC and IR Mechanisms.}
    Standard mechanism design arguments prescribe that any DSIC and IR mechanism has a distinctive structure: the payments are uniquely induced by the allocation region which, in turn, has to respect a monotonicity property (see \Cref{fig:approximation_challenges}).

    \begin{restatable}[Monotone Allocation Region \& Myerson Payments]{definition}{paymentsallocation}
    \label{def:monotone_allocation}
        An allocation region $A\subseteq[0,1]^2$ is monotone  if for any $x=(x_1,x_2) \in A$ and $y=(y_1,y_2) \in [0,1]^2$ with $x_1 \ge y_1$ and $x_2 \le y_2$, it holds that $y \in A$. A mechanism is monotone if its allocation region is monotone. Given a monotone allocation region, then the associated Myerson payments are defined as follows:
        \begin{align*}
            p(\vs,\vb) =& \ind{(\vs,\vb) \in A} \cdot \max\{x \in [0,1] \mid (x,\vb) \in A\} &&\forall\; (\vs,\vb) \in [0,1]^2\\
            q(\vs,\vb) = &\ind{(\vs,\vb) \in A}\cdot \min\{y \in [0,1] \mid (\vs,y) \in A\} &&\forall\; (\vs,\vb) \in [0,1]^2 
        \end{align*}
    \end{restatable}

    In words, for an allocation region to be monotone it should be ``closed'' in a north-west direction. That is, if a point $(\vs,\vb)$ is in $A$ then any point $(\vs',\vb)$ with $\vs' \leq \vs$ and any point $(\vs,\vb')$ with $\vb' \geq \vb$ should be in $A$ as well. The payments of a point $(\vs,\vb) \in A$, in turn, correspond to the ``east'' projection minus the ``south'' one onto the allocation boundary.  For the above definition to be well-posed, we require (without loss of generality) that all allocation regions are topologically closed.

    \begin{restatable}{proposition}{characterization}
    \label{thm:mechanisms}
        A mechanism for bilateral trade is dominant-strategy incentive compatible and individually rational if and only if its allocation region is monotone and uses Myerson payments.
    \end{restatable}


    \xhdr{Repeated Bilateral Trade.}
    We study the following problem: At each time step $t$, a new seller-buyer pair arrives, with private valuations $v^t = (\vs^t,\vb^t) \in [0,1]^2$. The learner/broker proposes a mechanism $M^t \in \cM$, and the agents bid (truthfully because of the DSIC property). The trade occurs according to $M^t$ and the bids of the agents, which then leave forever. The broker earns $\prof(M^t, v^t)$, the difference between the price paid \emph{by} the buyer and that paid \emph{to} the seller. For simplicity, $\prof_t(\cdot)$ denotes the profit induced by $v^t$ at time $t.$ Formally, we have the following learning protocol:
    
        \begin{algorithm}[h!]
    \begin{algorithmic}[ht]
    \For{time $t=1,2,\ldots$}
            \State a new pair of agents arrives with private valuations $v^t=(\vs^t,\vb^t) \in [0,1]^2$
            \State the learner declares a mechanism $M^t \in \cM$
            \State the trade takes place according to mechanism $M^t$ and the bids
            \State the learner gains $\prof(M^t, v^t) = \prof_t(M^t) \in [-1,1]$
        \EndFor
    \end{algorithmic}
    \caption*{\textbf{The Learning Protocol}}
    \end{algorithm}

    \xhdr{Smooth Distributions.} We assume that the private valuations of the agents are generated by a $\sigma$-smooth adversary. More precisely, an (oblivious) adversary generates up-front a sequence $\cS$ of independent $\sigma$-smooth random variables. A random variable is smooth if it is characterized by a density function that is not ``too peaked'', as formalized in the following definition. 
    \begin{definition}[Smooth Distributions]
    \label{def:smoothness}
        A distribution $\cD$ over $[0,1]^2$ is said to be $\sigma$-smooth, for some $\sigma \in (0,1]$, if the following inequality holds for any Borel set $A$ in $[0,1]^2$:
        \(            \Psub{v \sim \cD}{v \in A} \le \frac{\cL(A)}{\sigma},
        \)
        where $\cL$ denotes the Lebesgue measure. A random variable is $\sigma$-smooth if its distribution is $\sigma$-smooth.
    \end{definition}

    \xhdr{Regret.} The regret suffered by a learning algorithm $\cA$ onto an adversarial sequence of $\sigma$-smooth random variables $\cS$ is defined as 
    \begin{equation}
    \label{eq:regret_bilateral}
        R_T(\cA,\cS) = \sup_{M \in \cM} \sum_{t=1}^{T} \E{\prof(M, v^t) -  \prof(M^t, v^t)},
    \end{equation}
    where the expectation is with respect to the random valuations $v^1, v^2, \dots$ and possibly on the internal randomness of the algorithm. The regret of  $\cA$ is defined by its performance on the worst sequence of $\sigma$-smooth distributions: $R_T(\cA) = \sup_{\cS} R_T(\cA,\cS)$. 
    The typical goal is to achieve sublinear worst-case regret, the algorithm's average profit converges to that of the benchmark, on any input.

\section{The Grid Mechanisms}
\label{sec:grid}
    
    A central ingredient in our algorithmic construction resides on the concept of \emph{grid} mechanisms, i.e., mechanisms whose allocation region can be described in terms of a uniform grid on the $[0,1]^2$ square. While this family is already used in previous work \citep{AggarwalBDF24}, we are the first to prove the $L^1$-net guarantee needed in the chaining analysis. 
    The family of $\eps$-grid mechanisms is parameterized by a precision $\eps$ which we assume, for convenience, to be a power of $2$ ($\eps = 2^{-h}$, for some $h \in \mathbb{N}$). Consider the uniform grid $V_{\eps}$ of step-size $\eps$ of $[0,1]^2$, i.e., the points of the form $(k \eps, j\eps)$ for $k$ and $j$ in $0, \dots, \nicefrac{1}{\eps}$. The grid naturally partitions the $[0,1]^2$ square into squared ``tiles'' of the form $[k \eps, (k+1)\eps] \!\!\times \!\![j \eps, (j+1)\eps]$, we denote with $\cG_{\eps}$ the set of such tiles.
    \begin{definition}[Grid Mechanisms]
        The family of $\eps$-grid mechanisms $\cM_{\eps}\subseteq \cM$ is composed by all the mechanisms whose (monotone) allocation region is given by union of tiles in $\cG_{\eps}$.
    \end{definition}
    We slightly overload the definitions and require all allocation regions of mechanisms in $\cM$ (and thus also of grid mechanisms) to contain the segments $(0,0)-(0,1)$ and $(0,1)-(1,1)$. This technical assumption is compatible with monotonicity and is without loss of generality (by the smoothness assumption, adding zero-measure sets to an allocation region does not affect expected profit). The boundary of the allocation region of a generic $\eps$-grid mechanism is then given by the union of segments joining points in the underlying uniform grid, resulting in a path from $(0, 0)$ to $(1, 1)$. This allows bounding $|\cM_{\eps}|$ via a simple combinatorial argument.

        \begin{restatable}{lemma}{cardinality}
        \label{lem:card_bound}
            The cardinality of $\eps$-grid mechanisms respects ${2^{\nicefrac{1}{\eps}}} \le \lvert \cM_\eps\rvert \leq {(2e)^{\nicefrac{1}{\eps}}}$. 
        \end{restatable}
        \begin{proof}
        Consider the directed graph $G_{\eps}$ whose nodes are given by the uniform grid $V_{\eps}$, and whose edges are the sides of the tiles of $\cG_{\eps}$, pointing either from left to right (for horizontal edges) or bottom-up (for vertical ones). 
        Each $\eps$-simple mechanism is identified by a path on this graph going from $(0,0)$ to $(1,1)$, which --- together with the left and upper sides of $[0,1]^2$ --- encloses the allocation region. We then only need to count the number of paths of this type. Each path contains exactly $\nicefrac 2 \eps$ edges, equally divided into horizontal and vertical edges. All in all, we have $ |\cM_{\eps}| = \binom{\nicefrac 2 \eps}{\nicefrac 1 \eps} $, so that 
        \[
            2^{\nicefrac{1}{\eps}}\le |\cM_{\eps}| \le \left(2e\right)^{\nicefrac{1}{\eps}},    
        \]
        where we use the inequalities $\left(\nicefrac{n}{k}\right)^{k} \le \binom{n}{k} \le \left(\nicefrac{en}{k}\right)^{k}$. 
        \end{proof}
        
        The smoothness allows to nicely approximate the profit of any mechanism via $\eps$-grid mechanism, via some sort of $\eps$-grid ``inner'' approximation: given a mechanism $M$, just consider the union of all the tiles in $\cG_{\eps}$ which are fully contained in the allocation region of $M$ (see \Cref{fig:approximation_challenges,fig:approximations}).

        \begin{definition}[$\eps$-approximations]
        \label{def:approximations_prof}
        For any precision $\eps>0$ and $M \!\in\! \cM$ with allocation region $A$, we define $\app{\eps}{M} \!\in\! \cM_{\eps}$ as the mechanism whose allocation region is the union of all $P \!\in\! \cG_{\eps}$ s.t. $P \subseteq A$.
        \end{definition}

        If no tile is fully contained in the allocation region of a given mechanism, then the corresponding approximating mechanism is $M_{\emptyset}$, which only allocates in the union of the segments $(0,0)-(0,1)$ and $(0,1)-(1,1)$ (i.e., with zero probability). Clearly, due to the smoothness assumption, such mechanism has $0$ expected profit. For a generic mechanism $M$, the associated $\eps$-approximated mechanism is a good approximation to the expected profit of $M$. We have the following ``$L^1$-bound''. 
        \begin{figure}[t]
            \centering
            \includegraphics[width=0.9\linewidth]{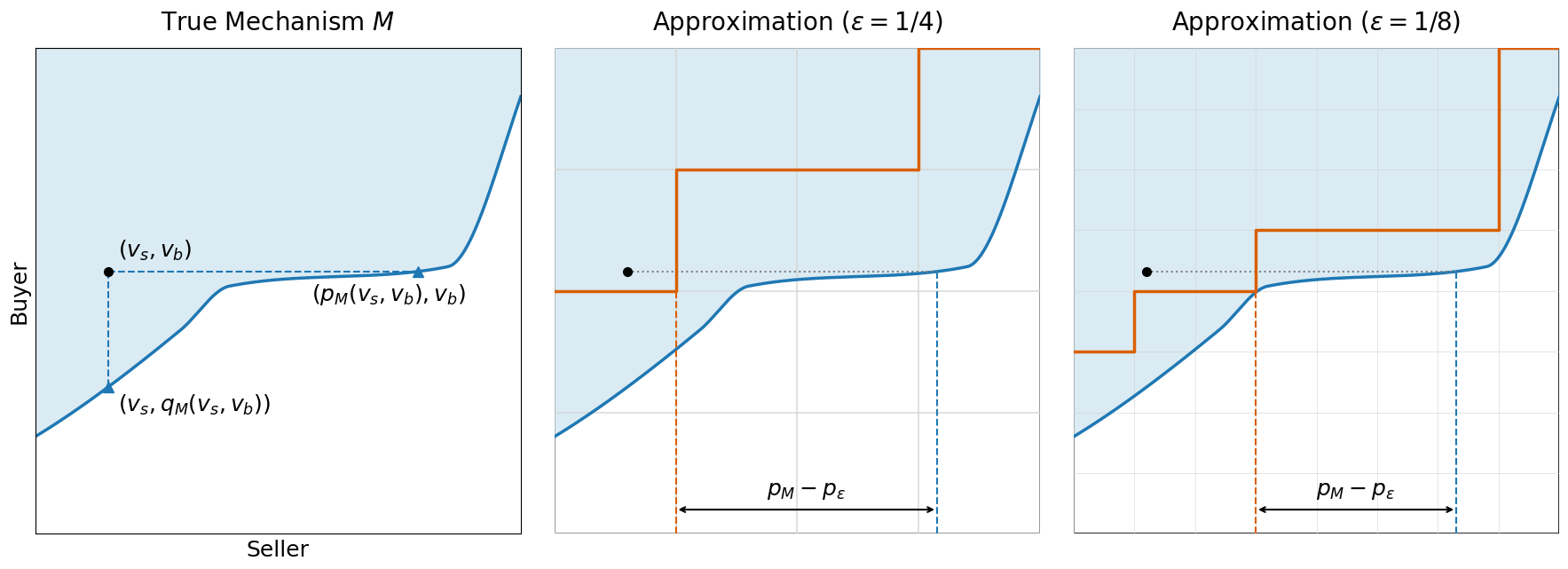}
            \caption{A monotone mechanism and two of its approximations. Notice the seller payment error.}
            \label{fig:approximation_challenges}
        \end{figure}

        \begin{figure}[t!]
        \centering
        \includegraphics[width=\linewidth]{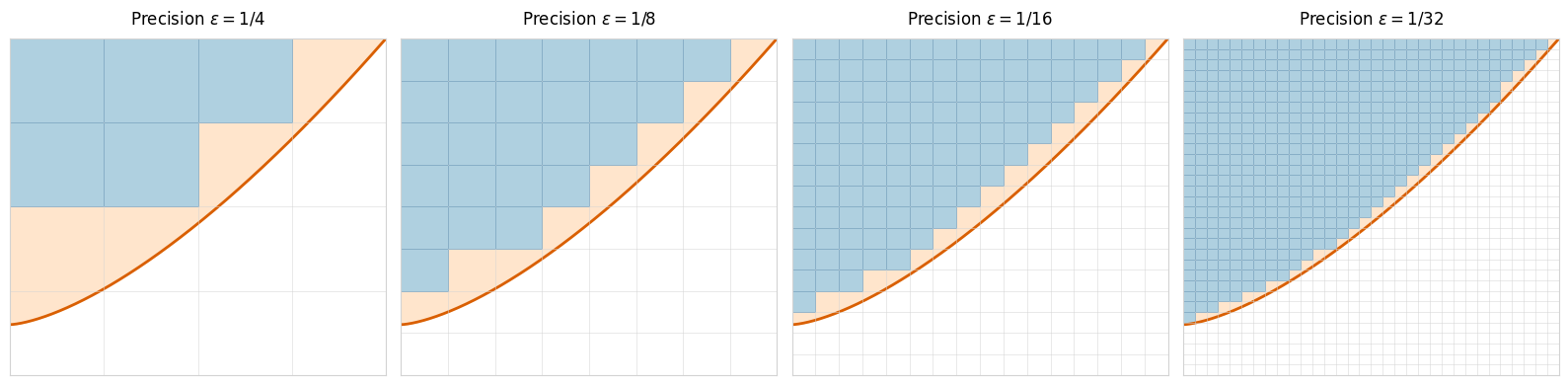}
        \caption{An illustration of the approximating map defined in \Cref{def:approximations_prof}, for different values of $\eps$. The approximation of the allocation region is in blue, the original mechanism has orange boundary.}
        \label{fig:approximations}
        \end{figure}

    \begin{restatable}{lemma}{lonenet}
    \label{lem:l1_net}
    Let $v$ be a $\sigma$-smooth random variable on $[0, 1]^2$. For every $\eps > 0$ we have the following inequality:
    \(
        \E{\lvert\prof(M, v)-\prof(\app{\eps}{M}, v)\rvert} \leq 6 \cdot \nicefrac{\eps}{\sigma}
    \).
    \end{restatable}
    \begin{proof}
        Consider a generic mechanism $M$ with allocation region $A_M$ and payments $p_M, q_M$, and its $\eps$-grid approximation $\app{\eps}{M}$, with allocation region $A_{\eps}$ and payments $p_{\eps}, q_{\eps}$. We can decompose the distance in profit between the two mechanisms depending on whether the realized valuation falls inside $A_{\eps}$ or in $A_M \setminus A_{\eps}$. In the latter case, the profit of $\app{\eps}{M}$ is zero outside its allocation region:
        \begin{equation}
        \label{eq:v_notin_Aphi}
            \E{\lvert\prof(M, v)\!-\!\prof(\app{\eps}{M}, v)\rvert\ind{v \in A_M \setminus A_{\eps}}} \!\le\! \P{v \in A_M \setminus A_{\eps}}\!\le\! \frac{2\eps}{\sigma}.
\end{equation}
    By a simple potential argument, there are indeed at most $\nicefrac{2}{\eps}$ tiles in $\cG_{\eps}$ that have non-trivial intersection with $A_M$ (i.e., they intersect with $A_M$ but are not fully contained in it). Moreover, each such tile has probability at most $\nicefrac{\eps^2}\sigma$ by smoothness, and their union is a superset of $A_M\setminus A_{\eps}$, so that $\P{v \in A_M\setminus A_{\eps}} \le \nicefrac{2\eps}{\sigma}$; see \Cref{fig:approximation_challenges}.
         The argument for $A_{\eps}$ is more subtle, as the payments induced by the grid mechanism may be considerably different than those computed by $M$, see again \Cref{fig:approximation_challenges}. Since the profit is composed by the price paid by the buyer, corresponding to the vertical projection, and that paid to the seller, horizontal projection, we partition $[0,1]^2$ into columns $C$ and rows $R$. A column is simply a rectangle of the type
         $(k\eps, (k+1)\eps] \times [0,1]$, while a row has the form $[0,1] \times (j\eps, (j+1)\eps]$. We can associate to each columns (respectively rows), an upper bound on the gap between the price paid by the buyer (respectively to the seller) in the two mechanisms:
        \(
            \Delta_c = \sup_{v \in c \cap A_\eps} \left\{q_{\eps}(v) - q_M(v)\right\}\) and \(\Delta_r = \sup_{v \in r\cap A_\eps} \left\{p_M(v) - p_{\eps}(v)\right\}.
        \)
        We then have the following chain of inequalities, using $p_\eps$ and $q_\eps$ in place of $p_{\app{\eps}{M}}$ and $q_{\app{\eps}{M}}$:
        \begin{align*}
            &\E{\lvert\prof(M, v)-\prof(\app{\eps}{M}, v)\rvert\ind{v \in A_{\eps}}} \\
            &\le \E{\left(|q_{\eps}(v) - q_M(v)| + |p_M(v) - p_{\eps}(v)|\right)\ind{v \in A_{\eps}}} \\
            &\leq \E{\sum_{c \in C}|q_{\eps}(v) - q_M(v)|\ind{v \in A_{\eps} \cap c} + \sum_{r \in R}|p_M(v) - p_{\eps}(v)|\ind{v \in A_{\eps} \cap r}}\\
            &\leq\sum_{c \in C}\Delta_c\P{v \in A_{\eps} \cap c} + \sum_{r \in R}\Delta_r\P{v \in A_{\eps} \cap r} \tag{By def. of $\Delta_r$ and $\Delta_c$}\\
            &\leq \frac{\eps}{\sigma}\left(\sum_{c \in C}\Delta_c + \sum_{r \in R}\Delta_r \right)\tag{By smoothness and $\cL(c)=\cL(r)={\eps}$} \le 4\frac{\eps}{\sigma},
        \end{align*}
        where the last inequality follows by arguing that both $\sum_c \Delta_c$ and $\sum_r \Delta_r$ are upper bounded by $2$. The formal proof is reported in \Cref{cl:bound_sum_error}.
    \end{proof} 
        \begin{claim}
        \label{cl:bound_sum_error}
            It holds that $\sum_{c\in C}\Delta_c \leq 2$ and $\sum_{r\in C}\Delta_r \leq 2$.
        \end{claim}    
        \begin{proof}   
        We prove the statement for the columns, the argument for the rows is analogous. 
        Any column $c$ can be written as $I_c \times [0,1]$, where $ I_c = (k\eps, (k+1)\eps]$, for some integer $k$. Provided that the trade happens (i.e. $v\in A_\eps \subseteq A_M$), the payment function $q_M$ is constant in the buyer valuation $v_b$ given the seller's $v_s$, and we focus without loss of generality on $v_b = 1$. By construction of the approximating mechanism, $q_\eps(\cdot,1)$ is constant in $I_c$ as a function of the seller valuation, and is equal to
        $$
        Q_c = \eps\left\lceil \frac{q_M((k+1)\eps, 1)}{\eps}\right\rceil.
        $$
        Notice that $Q_c$ is the height of the lower side of the lowest tile fully contained in both $c$ and $A_M$; if the intersection of the two does not contain any tile, then $Q_c = 1$. At this point, consider $\Delta_c$. Since $q_M$ is weakly increasing as a function of $v_s$ (a consequence of the monotonicity of $A_M$), the gap between $q_{\eps}(v_s, 1)$ and $q_M(v_s, 1)$ over the interval $(k\eps, (k+1)\eps]$ is maximized approaching the infimum of the interval. Therefore, taking the right-sided limit, we get $\Delta_c = Q_c - \lim_{v_s \rightarrow k\eps^+} q_M(v_s, 1)$.

        Assuming it exists, i.e. $k\neq 0$, let us now consider the previous column, denoted as $c_{-}$, corresponding to the interval $I_{c_{-}} = ((k-1)\eps, k\eps]$. Repeating the same analysis, we get $Q_{c_{-}} = \eps\left\lceil \frac{q_M(k\eps, 1)}{\eps}\right\rceil$. By definition of the ceiling function, we have $Q_{c_{-}}-\eps\leq q_M(k\eps, 1)$. Furthermore, by the monotonicity of $q_M$, its value at the boundary is bounded by its right limit: $q_M(k\eps, 1) \leq \lim_{v_s \rightarrow k\eps^+} q_M(v_s, 1)$. Together with the value of $\Delta_c$, we thus get:
        $$
        \Delta_c = Q_c-\lim_{v_s \rightarrow k\eps^+} q_M(v_s, 1)\leq Q_c - q_M(k\eps, 1) \leq Q_c - (Q_{c_{-}}-\eps), 
        $$
        which implies $Q_c-Q_{c_{-}} \geq \Delta_c -\eps$. If $k=0$ and $c_{-}$ does not exist, then $q_{\eps}(0, 1)$ is null by definition, with $Q_c \geq \Delta_c$ being trivially satisfied. 

        To conclude, since $q_{\eps}(\cdot, 1)$ is monotone and its image is a subset of $[0, 1]$, its total variation is bounded by $1$. Therefore, we get, with the convention of letting $Q_{c_{-}} = 0$ when $I_c = (0, \eps]$: 

        $$
        1 \geq \sum_{c\in C} (Q_c-Q_{c_{-}}) \geq \sum_{c\in C}(\Delta_c-\eps),
        $$ 

        which implies, using $\lvert C\rvert \le \nicefrac{1}{\eps}$: 
        \(
        \sum_{c\in C}\Delta_c \leq 2.
        \)
    \end{proof}

        

\section{Hierarchical Mechanism}

    
    In this section, we present our learning algorithm \hier and analyze its regret property. It instantiates multiple copies of the classical \hedge algorithm \citep{FreundS97} as subroutines; so we point to Appendix~\ref{app:hedge} for a quick recap of the algorithm and of its guarantees.

    \begin{algorithm}[t!]
    \begin{algorithmic}[ht]
    \State \textbf{Environment:} Repeated Bilateral Trade against a $\sigma$-smooth adversary on time horizon $T$
    \State Construct the mechanism tree $\cT$ with target precision $H = \nicefrac 12 \lceil  \log_2 T\rceil$
    \For{each level $h<H$ and mechanism $M \in \cM_h$}
        \State initialize a local version of \hedge $\cH^h(M)$, with actions $\cN_M$ and learning rate $\eta_h = \nicefrac{2^h}{\sqrt{T}}$
    \EndFor 
    \For{time $t=1,2,\ldots, T$}
        \For{each level $h<H$ and mechanism $M \in \cM_h$}
            \State Receive the prediction $\cH_t^h(M)$ from the local version of \hedge $\cH^h(M)$
        \EndFor
        \State Propose mechanism $M^t = \target{t}{M_\emptyset}$ and observe valuations $(\vs^t,\vb^t)$ 
        \For{each level $h<H$ and mechanism $M \in \cM_h$}
        \State Feed to $\cH^h(M)$ rewards $\Delta^h_t(M')= \nicefrac{1}{2}\cdot (\prof_t(\target{t}{M'}) - \prof_t(M))\, \forall \, M' \in \cN_M$
    \EndFor 
    \EndFor
    \end{algorithmic}
    \caption*{\hier}
    \end{algorithm}

    \xhdr{Mechanism Tree.} We introduce the mechanism tree $\cT$, whose nodes are the $2^{-h}$-grid mechanisms, for $h=0, \dots, H$ (see also \Cref{fig:hierarchical_tree}). For simplicity, we denote with $\cM_{h}$ the class of $2^{-h}$-grid mechanisms, and with $\app{h}{\cdot}$ (instead of $\app{2^{-h}}{\cdot}$) the generic $2^{-h}$ approximating mechanism in $\cM_h$. The $h$-th level of this tree is given by the nodes in $\cM_h$ (for simplicity, we ignore at all levels the suboptimal mechanism which always allocates, so that $\cM_0$ only contains the mechanism $M_{\emptyset}$ which almost surely never allocates). We specify the edges of $\cT$: there only exists edges going from one level to the next one, connecting the generic mechanism $M \in \cM_h$ with the unique $\app{h-1}{M} \in \cM_{h-1}$. For any node/mechanism $M$, we denote with $\cN_M$ its children. We say that a mechanism is an internal node if it is not a leaf. Note that by construction, for each node $M \in \cM_h$ there exists a unique root-to-node path in $\cT$, namely, $M_{\emptyset} = \app{0}{M}, \app{1}{M}, \dots, \app{h}{M}=M$.

    \xhdr{The algorithm: \hier.} Our algorithm, \hier, maintains a separate instantiation $\cH^h(M)$ of the \hedge algorithm for each internal node $M \in \cM_h$ in the tree, and uses them to decide which mechanism to propose at each time step $t$. We refer to the pseudocode for further details. For each local instantiation of \hedge to be well defined, we need to specify three things: the (a) learning rate, (b) the action space, and (c) the corresponding rewards. (a) As learning rate we use a level-dependent value $\eta_h= \nicefrac{2^{h}}{\sqrt T}$, while (b) the actions of the version of \hedge $\cH^h(M)$ supported on mechanism $M \in \cM_h$ are the children $\cN_M \subseteq \cM_{h+1}$ of $M$ in the mechanism tree $\cT$. Finally, (c) the rewards are defined recursively starting from the leaves and moving up: for each mechanism/children $M' \in \cN_M$, the reward of playing $M'$ is defined as $\Delta^h_t(M')=\nicefrac{1}{2}\cdot(\prof_t(\target{t}{M'}) - \prof_t(M))$, where the target function $\target{t}{\cdot}$ is also defined recursively ($\cH_t^h(M)$ denotes the random action played by the local hedge $\cH^h(M)$ at time $t$): 
        \[
            \begin{cases}
                \target{t}{M}=&M \quad  \text{for a leaf}\\
                \target{t}{M}=&\target{t}{\cH_t^h(M)} \quad \text{otherwise}
             \end{cases}
        \] 
    
\begin{figure}[t]
        \centering
        \includegraphics[width=0.85\linewidth]{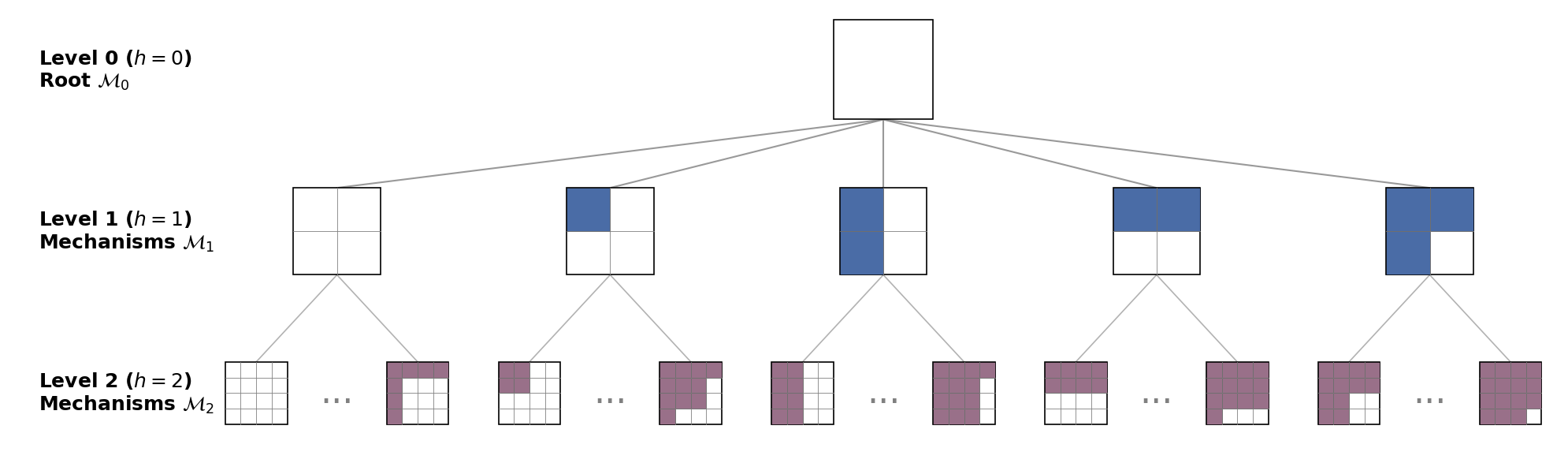}
        \caption{The first levels of the mechanism tree. We are only showing a subset of level $h=2$.}
        \label{fig:hierarchical_tree}
    \end{figure}
    Stated differently, the reward of choosing mechanism $M'$ in the local \hedge version in $M$ is given by the \emph{residual} profit achievable by following the choices of the local \hedge instantiations down the subtree rooted in $M$, up to the leaf $\target{t}{M'}$.
    Given all these local \hedge, the learner actually proposes to the agents the mechanism $M^t \!\!=\!\! \target{t}{M_{\emptyset}}$. Namely, at each time step, it follows the root-to-leaf path sampled by the local instantiations of \hedge, and chooses the resulting leaf mechanism. Once the pair $(\vs^t,\vb^t)$ is revealed, then \hier can update the internal state of each local \hedge.

        \paragraph{The analysis of \hier.} We denote with $M^t$ the random mechanism played by the algorithm at the generic time $t$, and with $\best$ the best mechanism on the instance\footnote{Although it is possible to prove with some work that such $\max$ is indeed attained, one can alternatively define $\best$ as any mechanism whose expected profit is optimal up to some small constant term, with this not changing the analysis.}. We use $\bestapp{h}$ to denote the generic approximating mechanism $\app{h}{\best}$, and with $\cH^{h}$ (instead of $\cH^h(\bestapp{h})$) the corresponding version of \hedge. In the analysis of \hier we only need to consider these local versions of \hedge, one for each level of the tree, whose properties are summarized in the following lemma. 
        \begin{lemma}[Local Regret]
        \label{lem:local}
        For each level $h < H$, we have the following bound on the regret of the version of \hedge $\cH^h$ at that level, where expectations are taken with respect to the algorithm randomization and the distribution of the valuations at each time step:
        \[
            \sum_{t=1}^T \E{\Delta^h_t({\bestapp{h+1}}) - \Delta^h_t(\cH_t^h) } \le \frac{7}{\sigma} \sqrt{T}.            \]            
        \end{lemma}
        \begin{proof}[Proof of \Cref{lem:local}]
            Consider the generic level $h < H$, the \hedge algorithm $\cH^h$ corresponds to node $\bestapp{h}$, and one of its actions is the next element $\bestapp{h+1}$ in the root-to-leaf path to the best mechanism $\best$. For brevity, denote with $\cN_{h+1}$ the children of $\bestapp{h}$.
            \begin{align*}
            \sum_{t=1}^T\E{  \Delta^h_t({\bestapp{h+1}})- \Delta^h_t(\cH_t^h)} \!&\le \!\! \max_{M \in \cN_{h+1}}\!\!  \sum_{t=1}^T\E{ \Delta^h_t(M) - \Delta^h_t(\cH_t^h)} \tag{As $\bestapp{h+1} \in \cN_{h+1}$}\\
            &\le \frac{\log |\cN_{h+1}|}{\eta_h} + \eta_h\cdot \sum_{t=1}^T\E{ \Delta^h_t(\cH_t^h)^2} \tag{\hedge guarantees}\\
            &\le \!4 \frac{ 2^h}{\eta_h}\!+\! \eta_h\cdot\!\sum_{t=1}^T \!\max_{M \in \cN_{h+1}}\E{ |\Delta^h_t(M)|}\tag{\Cref{lem:card_bound}, $|\Delta^h_t(M)|\le 1$}
            \end{align*}
            Note, the second inequality follows by the standard analysis of the \hedge algorithm; for completeness, we report a formal proof in \Cref{app:hedge}). In the third inequality, we are using that the distribution of \hedge at time $t$ is independent of the valuations at time $t$. Now, let $M\in \cM_h$ and recall the definition of $\Delta^h_t(M') = \nicefrac{1}{2}\cdot(\prof_t(\target{t}{M'}) - \prof_t(M))$ for any $M' \in \cN_M$. By construction, $\target{t}{M'}$ is a leaf in the sub-tree rooted in $M$, therefore $\app{h}{\target{t}{M'}} = M$. 
            This implies that we can apply \Cref{lem:l1_net} and upper bound $\max_{M \in \cN_{h+1}}\E{ |\Delta^h_t(M)|}$ with $\nicefrac{3}{\sigma} \cdot 2^{-h}$.
            Plugging in this inequality in the previous display yields: 
            \[
            \sum_{t=1}^T\E{ \Delta^h_t({\bestapp{h+1}}) - \Delta^h_t(\cH_t^h)} \le 4 \cdot\frac{2^h}{\eta_h}+ 3\cdot \frac{\eta_h}{\sigma}2^{-h} \cdot T \le \frac{7}{\sigma} \sqrt{T},  
            \]
            where the last inequality follows by our choice of the level-dependent learning rate. 
        \end{proof}

        The lemma ensures that the local regret suffered along each edge of the root-to-leaf path to $\bestapp{H}$ enjoys the same rate! Our choice of the learning rate manages to find the right trade-off between the number of available mechanisms (higher $h$ corresponds to larger number of children) and the expected magnitude of the residual profit (higher $h$ corresponds to lower magnitude, by \Cref{lem:l1_net}). We then relate the actual profit of \hier with the benchmark $\best$ (more precisely to its approximation $\bestapp{H}$). This is the core of the argument, and is analogous to what happens in probabilistic chaining. 
        \begin{lemma}[Chaining the Profit]
        \label{lem:chaining}
            For each level $0<h \le H$, we have the following regret bound:
            \(
                \sum_{t=1}^T \E{\prof_t(\target{t}{\bestapp{h}}) - \prof_t(M^t)} \le \frac{14}{\sigma} h \sqrt T.
            \)
        \end{lemma}
        \begin{proof}
            We prove the result by induction on the level $h \le H$. From \Cref{lem:local} applied at the version of \hedge corresponding to the root, we have: 
            \(
            \sum_{t=1}^T\E{ \Delta^0_t({\bestapp{1}}) - \Delta^0_t(\cH_t^0)} \le \frac{7}{\sigma}\sqrt T.
            \)
            For any $M\in \cN_{M_{\emptyset}}$, $\Delta_t^0(M) = \nicefrac{1}{2}\cdot \prof_t(\target{t}{M})$ almost surely, as the profit of the root mechanism $M_\emptyset$ is almost surely $0$. Moreover, by design of the algorithm, it holds that $2\Delta^0_t(\cH_t^0) = \prof_t(M^t)$, while $2\Delta^0_t({\bestapp{1}})=\prof_t(\target{t}{\bestapp{1}})$, since $\bestapp{1}\in \cN_{M_\emptyset}$. Combining these facts yields the base case:
            \[
                \sum_{t=1}^T\E{\prof_t(\target{t}{\bestapp{1}})- \prof_t(M^t)}\le \frac{14}{\sigma}\sqrt T.
            \]
            Assume now by inductive hypothesis that the statement holds for a generic level $h<H$, it means that, recalling that $\cH^h$ refers to the \hedge instance for node $\bestapp{h}$: 
            \begin{align*}
                \sum_{t=1}^T \E{\prof_t(M^t)} &\ge - \frac{14}{\sigma} h \sqrt T + \sum_{t=1}^T \E{\prof_t(\target{t}{\bestapp{h})} \pm \prof_t(\target{t}{\bestapp{h+1}}) \pm \prof_t(\bestapp{h}))}\\
                &{=} - \frac{14}{\sigma} h \sqrt T + \sum_{t=1}^T \E{\prof_t(\target{t}{\bestapp{h+1})}}
                + 2\sum_{t=1}^T \E{\Delta^h_t(\cH_t^h) - \Delta^h_t(\bestapp{h+1}))}\\
                &\ge - \frac{14}{\sigma} (h+1) \sqrt T + \sum_{t=1}^T \E{\prof_t(\target{t}{\bestapp{h+1})}},
            \end{align*}
            where the equality follows by the definition of $\Delta_t^h$, and the last inequality by \Cref{lem:local}.
        \end{proof}
        \begin{theorem}
        \label{thm:hierHedge}
            Consider the problem of profit maximization in repeated bilateral trade against a $\sigma$-smooth adversary. We have the following regret bound:
            \(
                R_T(\hier) \le \frac{20}\sigma\sqrt{T}\log T.
            \)
        \end{theorem}
        \begin{proof}
            Recall that $\bestapp{H}$ is the approximating mechanism with precision $2^{-H}$ of the optimal mechanism $\best$. We have the following chain of inequalities, which concludes the proof as $H = \nicefrac{1}{2}\lceil \log_2 T\rceil$:
            \begin{align*}
                \frac{14}{\sigma} H \sqrt T &\ge \sum_{t=1}^T \E{\prof_t(\target{t}{\bestapp{H}}) - \prof_t(M^t)} \tag{By \Cref{lem:chaining} on the last level $H$}\\
                 &=\sum_{t=1}^T \E{\prof_t({\bestapp{H}}) - \prof_t(M^t) \pm \prof_t(\best)} \tag{As $\bestapp{H}$ is a leaf}\\
                 &\ge \sum_{t=1}^T \E{\prof_t({\best}) - \prof_t(M^t)} - \frac{6}{\sigma}\cdot 2^{-H} T\tag{By \Cref{lem:l1_net}} \\
                 & \geq \sum_{t=1}^T \E{\prof_t({\best}) - \prof_t(M^t)} - \frac{6}{\sigma}\cdot \sqrt{T}.\qedhere
            \end{align*}
        \end{proof}


\section{A Reduction from Joint Ads to Bilateral Trade}

    In this section, we provide an high-level intuition behind our reduction from joint ads to profit maximization in bilateral trade. We refer to Appendix~\ref{subsec:joint_ads} for a definition of the joint ads model, while the formal proof of the reduction is in Appendix~\ref{app:reduction}.

    \xhdr{Joint ads.}
    Here, two buyers cooperate in an auction for a non-excludable good: either they both get it, or none does. The mechanism designer wants to maximize the profit/revenue extracted from the trade, while enforcing DSIC and IR constraints. The main difference from bilateral trade is the ``direction of monotonicity'': allocation regions need to be \enquote{north-east} closed, while payments are given by the \enquote{south} and \enquote{west} projections\footnote{Recall: in bilateral trade, the monotonicity is in the north-west direction.}. For a visualization, we refer to the second picture of \Cref{fig:reduction}. The learning protocol is analogous to bilateral trade. The agents are characterized by private valuations $(v^t_1, v^t_2)$, drawn from a $\sigma$-smooth distribution. The learner proposes mechanism $M^t$ with allocation $A^t$ and two payments $p^t_{1}$ and $p^t_2$, the trade happens if $(v^t_1,v^t_2)\in A^t$ and the revenue is:
            \(        
                \rev_t(M^t) = p^t_1(v_1,v_2) + p^t_2(v_1,v_2).
            \)
    While the two problems share some similarities, in bilateral trade we may have negative $\prof$ for some trades, while this cannot happen for joint ads. In both problems, the trade {not} happening results in $0$ profit/revenue, thus introducing an asymmetry.

    \xhdr{Our Reduction.} To overcome this issue, we restrict the bilateral trade instance to $[0,\nicefrac 12]\times[\nicefrac{1}{2},1]$, so that no trade results in \emph{negative} payment (as the buyer pays at least $\nicefrac 12$, and the seller receives at most $\nicefrac 12$). Luckily, by monotonicity this restriction goes ``in the right'' direction for the reduction to happen. Given a learning algorithm $\cAbl$ for bilateral trade, we now show how to derive $\cAja$ for joint ads: at time $t$, we query mechanism $\blM^t$ from $\cAbl$, we then \emph{clip} it to the square of interest $[0,\nicefrac 12]\times[\nicefrac 12,1]$ (see first plot in \Cref{fig:reduction}), and apply the {inverse of the} following map $f:(x,y) \to (\nicefrac 12\cdot(1-x),\nicefrac 12\cdot(1+y))$. The resulting allocation region $\jaA^t$ induces the mechanism $\jaM^t$. The joint ads learner proposes  $\jaM^t$ and observes the agents valuations $(v_1^t,v_2^t)$, which are then mapped back to bilateral trade {via $f$}. 
    As corollaries of the reduction, we prove in Appendix~\ref{app:reduction} the following:
    \begin{figure}[t]
    \centering
    \includegraphics[width=0.8\linewidth]{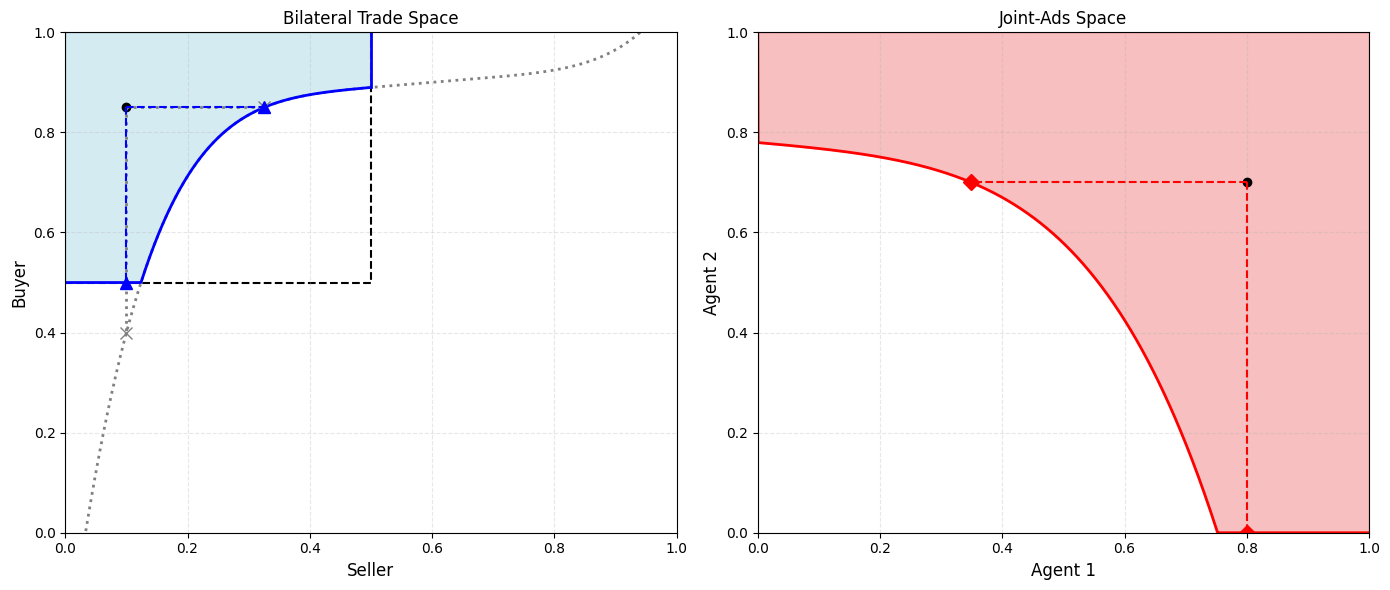}
    \caption{A visualization of our reduction from the joint ads problem.} 
    \label{fig:reduction}
    \end{figure}

    \begin{restatable}[A lower bound for bilateral trade]{corollary}{lowerboundbl}
    \label{cor:lower_bilateral}
        Consider the problem of profit maximization in repeated bilateral trade against a $\sigma$-smooth adversary, for $\sigma \in (0, \nicefrac{1}{48})$. Then, for any learning algorithm $\cA$, the following lower bound holds: 
        \(
        R_T(\cA) \geq \frac{3}{128}\sqrt{T}.
        \)
    \end{restatable}
    \begin{restatable}[An upper bound for joint ads]{corollary}{upperboundja}
    \label{cor:upper_joint}
    Consider the problem of revenue maximization in repeated joint ads against a $\sigma$-smooth adversary. There exists a learning algorithm which suffers regret at most $\frac{160}{\sigma}\sqrt{T}\log T$.
    \end{restatable}

\section{Conclusion and Future Directions}

    In this paper, we characterize the minimax regret rates for profit maximization in bilateral trade for the $\sigma$-smooth adversary, up to lower order terms.  We achieve this result via an algorithmic chaining approach, for a suitable $L^1$-net construction on the mechanism space. To the best of our knowledge, we are the first to use algorithmic chaining beyond online regression, and we envision it has a vast potential in economically motivated learning problems. We further formalize a reduction from the joint ads problem, a related mechanism design task. This allows us to provide the same tight results also for this new problem in the smooth adversarial setting. 
    We leave open two natural and exciting research questions that stem from our results: understanding the right dependency on the smoothness parameter $\sigma$ in the regret rate, and assessing whether the same rates can be attained {in time polynomial in $T$, as opposed to polynomial in the number of experts maintained.}
    
\section*{Acknowledgments}

    The work of SDG and FF was supported in part by the MUR PRIN grant 2022EKNE5K (Learning in Markets and Society), while CS was supported by a Google Research Award.

\bibliographystyle{plainnat}
\bibliography{references}


\appendix

\newpage
\section[On the Impossibility of a Deterministic Net]{On the Impossibility of a Deterministic $L^{\infty}$-net}
\label{app:fat-shattering}

    In this appendix, we comment on the challenges arising in applying directly the analysis from \citet{Cesa-BianchiGGG17} to our problem. In particular, to apply such results, we would need a finite family of mechanisms $\cM^\infty_\eps$ satisfying the following $L^\infty$-net guarantee: for any mechanism $M \in \cM$ there should exist an $M_{\eps} \in \cM_{\eps}^{\infty}$ such that

\begin{equation}
\label{eq:net_infty}
\lvert \prof(M, v)-\prof(M_{\eps}, v)\rvert \leq \eps   \quad \,\forall\, v\in [0, 1]^2
\end{equation}

This net property is unattainable due to the strongly discontinuous nature of the profit function. For any  $M_{\eps}$ \emph{different than} $M$, the symmetric difference of the allocation regions contains at least one valuation $v$ for which one mechanism extracts possibly non-null profit and the other one misses the trade. 

We have two ways of overcoming this impossibility: either focusing on a still-deterministic but weaker notion of $L^{p}$ bound, or by requiring a less stringent \emph{in-expectation bound} using smoothness. The latter is what we derive in the main body (we restate below the relevant lemma for convenience):

\lonenet*

We devote the rest of the appendix to argue that the first solution proposed, namely a weaker deterministic $L^p$ bound is hopeless (to some extent, this is reminiscent of the non-constructive nets used in \citet{HaghtalabRS24}). Such notion would read as follows. For any fixed choice of $n$ points $S$ in $[0,1]^2$ there exists a sample-dependent $\mathcal{M}_\varepsilon^p(S)$ such that for every $M \in \cM$ there exist an $M_{\eps}^p \in \mathcal{M}_\varepsilon^p(S)$ verifying:
\[
    \left(\frac1n \sum_{v\in S} \lvert \prof(M, v)-\prof(M^p_{\eps}, v)\rvert^p\right)^{\nicefrac 1p}\leq \varepsilon.
\]
Note, this notion is weaker than \Cref{eq:net_infty}. Standard combinatorial arguments tell us that one can always find a finite $\cM_{\eps}^p(S)$ net respecting the above property, whose cardinality is, however, \emph{exponential} in $|S|.$ Such trivial bound does not imply any non-trivial learning guarantees. 

Unfortunately, it turns out that it is impossible to get a sub-exponential net. Indeed, covering/packing duality (see Theorem 12.1 in \citet{bartlett_neural}) combined with the fact that the fat-shattering dimension is unbounded implies 
the non-existence of better-than-exponential $L^p$-nets.


We now prove that the fat-shattering dimension of the problem is unbounded. 
The fat shattering dimension~\citep[e.g. Chapter 11 in][]{bartlett_neural} is a \enquote{scale-sensitive} generalization of the VC dimension from \citet{vapnik1971uniform} from binary classification. It is a crucial tool to \textit{characterize the size of nets} for real-valued functions\footnote{An in-depth discussion about covering and packing numbers is out of scope for this paper, but we refer to \citet{bartlett_neural} for a good overview.}, and it is much more general than the stricter pseudodimension~\citep{pollard1984convergence}, which is unbounded (and thus useless) even for standard classes of functions (like the ones with bounded variation).

\begin{definition}[$\gamma$-fat shattering]
\label{def:fat_shat}
Let $\mathcal F$ be a class of functions mapping from $\mathcal X$ to $\mathbb{R}$, and let $\gamma > 0$. A set of points $S = \{x_1, x_2, \ldots, x_m\} \subset X$ is said to be $\gamma$-fat shattered by $\mathcal F$ if there exists a witness vector $s = (s_1, s_2, \ldots, s_m) \in \mathbb{R}^m$ such that for every Boolean assignment $y \in \{-1, 1\}^m$, there exists a function $f \in \mathcal F$ satisfying, for all $i\in [m]$:

\begin{align*}
    &f(x_i) \ge s_i + \gamma\quad \text{if $y_i = 1$}\\
    &f(x_i) \le s_i - \gamma \quad \text{if $y_i = {-}1$}
\end{align*}

The $\gamma$-fat shattering dimension of $\mathcal F$, $\fat{\mathcal F}{\gamma}$ is the cardinality of the largest $A\subseteq X$ which can be $\gamma$-fat shattered. If no upper bound can be found, the fat shattering dimension is defined to be infinite.
\end{definition}



\begin{proposition}[Unbounded fat shattering]
\label{prop:fat_shatter}
        For any $\gamma \leq \nicefrac{1}{4}$, the $\gamma$-fat shattering dimension of the profit functions in bilateral trade $\mathcal F_{\cM}$ is unbounded.
\end{proposition}
\begin{proof}
    Fix a arbitrary integer $m$ and fix any set $S = \{v_i\}_{i=1}^m \subset C = \left\{(x, y): x = y-\frac{1}{2}\right\}\cap [0, 1]^2$. Any valuation vector in $S$ lies on $C$, meaning that the difference between its entries is always $\nicefrac{1}{2}$.
    
     Let us fix an arbitrary vector of signs $y$: this induces a partition of $S$ into $S'$ and $S\setminus S'$. We can build a mechanism $M'$ and a monotone allocation region $A_{M'}$ such that the corresponding profit function $\prof(M', \cdot) \in \mathcal F_\cM$ assigns $\nicefrac{1}{2}$ to points in $S'$ and $0$ otherwise. More precisely, let $(x_i, y_i)=v_i$ and $R_i = [0, x_i]\times[y_i, 1]$ for any $i\in [m]$; we define $A_{M'}=\cup_{i:\, v_i \in S'} R_i$. 
     
     If we take $s \in \mathbb{R}^m$ to be the constant vector with entries equal to $\nicefrac 14$, we immediately have that $\prof(M', v_i)-s_i\geq \nicefrac{1}{4}$ if $y_i=1$ and $\prof(M', v_i)-s_i\leq -\nicefrac{1}{4}$ if $y_i=-1$. No assumption was made on $y$, so the same holds for any choice of $y$. Therefore, $S$ is $\nicefrac{1}{4}$-fat shattered by $\mathcal F$. The procedure is valid $\forall \, m \in \mathbb{N}$: $\fat{\mathcal F_{\mathcal M}}{\nicefrac{1}{4}}$ is thus unbounded. The claim then follows by noting that $\fat{\mathcal F}{\gamma}$ is monotone decreasing in $\gamma$, for any function class $\mathcal F$.
\end{proof}

\section{Mechanism Design Background}
\label{app:agt}

In this appendix, we recall some of the basic notions from mechanism design about bilateral trade and the joint-ads problem. We focus particularly on the known geometric characterizations of DSIC and IR mechanisms utilized throughout our work.

\subsection{Bilateral Trade}
\label{app:bilateral_trade}

In the static bilateral trade problem, a seller and a buyer interact to exchange a single good. The seller's and buyer's  (private) valuations for the good are $\vs \in [0,1]$ and $\vb \in [0,1]$, respectively. Seller and buyer submit not necessarily truthful bids $\bs \in [0,1]$ and $\bb \in [0,1]$ to a broker running mechanism $M$. A mechanism is characterized by an allocation region $A \subseteq[0,1]^2$, and pricing rules $p ,q:[0,1]^2 \to [0,1]$. 

The trade happens if and only if the bids $(\bs,\bb)$ belong to the allocation region $A$, while the payments are made according to $p$ and $q$. Payments are from the broker to the seller, and from the buyer to the broker: the seller thus receives $p(\bs,\bb)$, while the buyer pays $q(\bs,\bb)$.

The buyer's and seller's utility with valuations $\vb$ and $\vs$ are, when bidding $(\bb, \bs) \in [0,1]^2$ : 
\begin{align*}
\label{eq:utilities}
\util{\bs}{\bb}{s} &= \vs - \ind{(\bs,\bb)\in A} \cdot \vs + p(\bs, \bb)\\ \util{\bs}{\bb}{b} &= \ind{(\bs,\bb)\in A} \cdot \vb - q(\bs, \bb) \notag
\end{align*}
 We consider dominant-strategy incentive compatible (DSIC) and individually rational (IR) mechanisms. With this restriction, each player maximizes his utility with a truthful bid regardless of the other player's bid, and the utility from participating in the mechanism is at least as high as that from not participating in the mechanism. In formulas,
\begin{align*}
    \text{DSIC:}\quad \util{\vs}{\bb}{s} &\geq \util{\bs}{\bb}{s} &&\forall \ \vs \in [0,1], (\bs, \bb) \in [0, 1]^2 \\
    \util{\bs}{\vb}{b} & \geq \util{\bs}{\bb}{b} &&\forall \ \vb \in [0,1], (\bs, \bb) \in [0, 1]^2 \nonumber \\
    \text{IR}:\quad\;\;\; \util{\vs}{\bb}{s} &\geq \vs, \,\, \util{\bs}{\vb}{b} \geq 0 &&\forall \ (\vs, \vb) \in [0, 1]^2, (\bs, \bb) \in [0, 1]^2 \nonumber
\end{align*}

As a result of the aforementioned requirements, we obtain \Cref{thm:mechanisms}. A comprehensive proof of the statement, utilizing standard tools from \citet{Myerson81,MyersonS83}, can be found in \citet{DiGregorioDFS25}.
\paymentsallocation*
\characterization*

As stated in the main body, we denote the family of mechanisms characterized by the above proposition as $\cM$. Notice that the DSIC and IR conditions allow $p(v_s, v_b)\geq 0$ and $q(v_s, v_b)\leq 0$ when there is no trade\footnote{$p(v_s, v_b)< 0$ and $q(v_s, v_b) > 0$ would violate the IR constraint.}, which is more general than \Cref{def:monotone_allocation}. Specifically, the constraints only define payments up to a constant offset. However, considering our objective of profit maximization, having a strictly positive and strictly negative offset for the buyer’s and seller’s payments, respectively, would only deteriorate our objective. Therefore, we can restrict our focus to our definition of payments without any loss of generality.

\subsection{The Joint-Ads Problem}
\label{subsec:joint_ads}

\begin{figure}[!ht]
    \centering
    \includegraphics[width=0.8\linewidth]{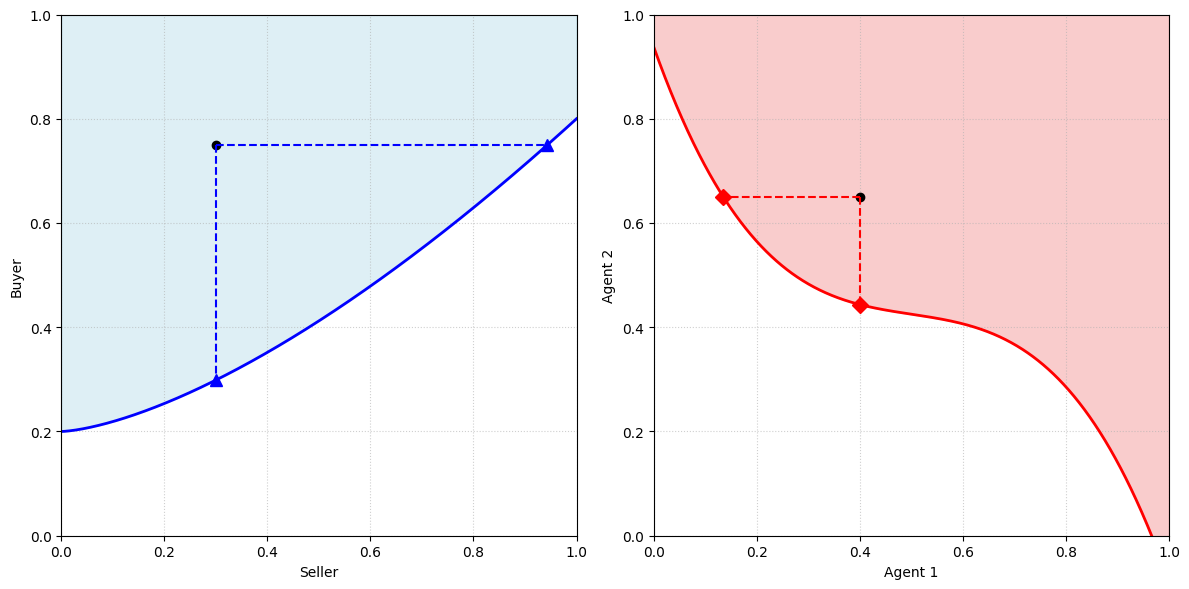}
    \caption{One bilateral trade mechanism (left) and one joint-ads mechanism (right), each with an allocated point and associated payments.}
    \label{fig:ja_bl_mechanisms}
\end{figure}

In the joint-ads problem, two buyers participate in an auction for a non-excludable\footnote{Either both agents get the good, or none of the two gets it.} specific good, and the goal of the mechanism designer is to maximize the revenue extracted from the trade.

This pair of agents is characterized by private valuations $(v_1, v_2)$ in the $[0,1]^2$ square. The learner proposes a mechanism $M$ to the agents, who then declare bids $(b_1,b_2)$, with a mechanism being composed of an allocation region $A\subseteq [0, 1]^2$ and two payment functions $p_1,p_2 : [0,1]^2 \to [0,1]$ that map bids to payments. The trade happens if $(b_1^t,b_2^t)\in A$ and the revenue of mechanism $M$ when the bids are $b_1$ and $b_2$ is defined as
        \[        
            \rev(M,b_1,b_2) = p_1(b_1,b_2) + p_2(b_1,b_2).
        \]

Clearly, the allocation function and the payment function cannot be considered arbitrary. As in the bilateral trade problem, the agents behave strategically, meaning that they strive to maximize their own quasi-linear utility $u_i(b_1,b_2) = v_i \cdot \ind{(b_1,b_2)\in A} - p_i(b_1,b_2)$ for $i \in \{1,2\}$. Note that, differently from the bilateral trade problem, here the two agents are symmetric, as they are both interested in purchasing a good. We require any proposed mechanism $M$ to be DSIC and IR. 

The family of mechanisms $\cM^{JA}$ that respect these two properties can be easily characterized with the same techniques used for bilateral trade (and for any single-parameter mechanism design problem). This tells us that $(i)$ an allocation rule is implementable if and only if it is monotone and $(ii)$ for each such allocation rule there exists a unique Myerson payment rule that makes it a DSIC and IR mechanism. For this reason, we only talk about valuations in the rest of the appendix, as the agents are incentivized to behave truthfully.

Given that the agents are symmetric, the meaning of monotonicity and the definitions of payments are different from the ones in bilateral trade; in joint-ads, an allocation region is monotone if it is closed in the \enquote{north-east} direction, while payments are given by the \enquote{west} and \enquote{south} projection to the boundary of the allocation region. See \Cref{fig:ja_bl_mechanisms} for a visualization.

\begin{definition}[Monotonicity and Payments for Joint-Ads]
\label{def:monotone_allocation_ja}
    An allocation region $A\subseteq[0,1]^2$ is monotone for joint-ads if, for any $x=(x_1,x_2) \in A$ and $y=(y_1,y_2) \in [0,1]^2$ with $x_1 \leq y_1$ and $x_2 \leq y_2$, it holds that $y \in A$. Moreover, for any such monotone region $A$, Myerson payments are defined as follows:
    \begin{align*}
        p_1(v_1,v_2) =& \ind{(v_1,v_2) \in A} \cdot \min\{x \in [0,1] \mid (x,v_2) \in A\} &&\forall\; (v_1,v_2) \in [0,1]^2\\
        p_2(v_1,v_2) = &\ind{(v_1,v_2) \in A}\cdot \min\{y \in [0,1] \mid (v_1,y) \in A\} &&\forall\; (v_1,v_2) \in [0,1]^2 
    \end{align*}
\end{definition}

\section{The Hedge Algorithm} 
\label{app:hedge}

    An important building block of our analysis is the \hedge algorithm. For completeness, in this appendix we provide a brief description of the algorithm and the regret bound it achieves.
    
    Consider the standard prediction with expert advice setting, on action set $A$, with possibly randomized reward functions $r_t: A \to [-1,1]$ generated upfront by an oblivious adversary. The \hedge algorithm with learning rate $\eta$ maintains a weight $w_t(a)$ for each action, defined as $w_t(a) = \exp\left(\eta \cdot \sum_{s=1}^{t-1} r_s(a)\right)$, and  plays actions $a_t$ by sampling proportionally to such weight. More precisely, the probability of choosing action $a$ at time $t$ is equal to $\pi_t(a) = \nicefrac{w_t(a)}{W_t},$ where $W_t = \sum_{a'} w_t(a')$. 
    In the main body we use the following, standard bound, that we reprove here for completeness \citep[for instance, see e.g., Chapter 5]{Slivkins19}. 

    \begin{lemma}
    \label{lem:hedge}
    Consider the standard prediction with expert advice setting. The \hedge algorithm with learning rate $\eta$ enjoys the following regret bound:
    \[
        \max_{a \in A} \sum_{t=1}^T\E{ r_t(a) - r_t(a_t)} \le \frac{\log |A|}{\eta} + \eta\cdot \E{\sum_{t=1}^T r_t(a_t)^2},
    \] 
    where the expectations are only with respect to the sampling of \hedge, with the inequality holding for every realization of the possibly randomized reward functions.
    \end{lemma}
    \begin{proof}
            The proof follows the standard pipeline, we just need to be a bit more careful about the scale of the rewards. The identity of the actual best action $a^\star$ only depends on the realized reward functions, and not on the actions chosen by the algorithm. Denote with $W_t$ denote the sum of the weights of all the actions at the beginning of time $t$. We have: 
            \begin{align}
                \log(W_{T+1}) &\ge \log(w_{T+1}(a^\star)) \tag{The weights are non-negative}\\
                &= \eta \sum_{t=1}^T r_t(a^\star).\label{eq:W_0}
            \end{align}
            Consider now the generic time step $t$, we have:
            \begin{align}
                \log \frac{W_{t+1}}{W_t} &= \log \sum_{a \in A} \frac{e^{\eta \sum_{s=1}^t r_s(a)}}{W_t} \tag{By definition of $W_{t+1}$}\\
                &= \log \sum_{a \in A} \pi_t(a) e^{\eta \cdot r_t(a)} \tag{By design of the algorithm}\\
                &\le \log \sum_{a \in A} \pi_t(a) \left(1 + \eta \!\cdot \! r_t(a) + \eta^2r_t^2(a)\right) \tag{$e^x \le 1+x+x^2, \, \forall x<1.79$}\\
                &= \log \left(1 + \eta \sum_{a \in A} \pi_t(a) r_t(a) + \eta^2 \sum_{a \in A} \pi_t(a) r^2_t(a)\right) \notag \\
                &\le \eta \sum_{a \in A} \pi_t(a) r_t(a) + \eta^2 \sum_{a \in A} \pi_t(a) r^2_t(a), \label{eq:w_t+1/w_t}
            \end{align}
            where the last inequality follows because $\log (1+x) \le x$ for all $x>-1.$ 
            \begin{align}
                \log \frac{W_{T+1}}{W_1} &=  \log \prod_{t=1}^T \frac{W_{t+1}}{W_t}\tag{Telescopic argument}\\
                &= \sum_{t=1}^T \log \frac{W_{t+1}}{W_t} \notag\\
                &\le \eta \sum_{t=1}^T \left[ \sum_{a \in A} p_t(a) r_t(a) + \eta \sum_{a \in A} p_t(a) r^2_t(a)\right]
                \tag{By \Cref{eq:w_t+1/w_t}} 
            \end{align}
        Noticing that $\log W_1 = \log \lvert A\rvert$, we get the claim.
        \end{proof}

\section{From Joint-Ads to Bilateral Trade: Reduction and Corollaries}
\label{app:reduction}

    In this Appendix, we formally reduce the joint-ads problem to bilateral trade. In particular, in \Cref{thm:reduction} we argue that any learning algorithm for (profit maximization in) bilateral trade  achieving some regret bound $\rho(T,\sigma)$ can be used to derive a learning algorithm for joint-ads suffering $O(\rho(T,\sigma))$ regret. 

    The main ingredient for the reduction is simple: an affine transformation $f:[0,1]^2 \to [0,\nicefrac 12] \times [\nicefrac 12,1] $ defined as follows:
    \[
        f(x,y) = \left(\frac{1-x}{2},\frac{1+y}{2}\right).
    \]
    It is immediate to see that $f$ is bijective from $[0,1]^2$ to $[0,\nicefrac 12] \times [\nicefrac 12,1]$. 

    \begin{theorem}[A Reduction from Joint-Ads to Bilateral Trade] 
    \label{thm:reduction}
    Let $\cAbl$ be a learning algorithm for profit-maximization in bilateral trade against the $\sigma$-smooth adversary which enjoys a regret bound $\rho(\sigma,T)$. Then we can construct an algorithm $\cAja$ for joint-ads against the $\sigma'$-smooth adversary which enjoys a regret bound $O(\rho(\sigma,T))$, where {$\sigma=\nicefrac{\sigma'}{4}$.} 
    \end{theorem}
    \begin{proof}
        Given a learning algorithm $\cAbl$ for bilateral trade, we  show how to derive $\cAja$ for joint ads which enjoys the same regret rate. We construct $\cAja$ by running $\cAbl$ on a suitable instance of bilateral trade. The reduction proceeds in two stages at each time steps. 
        
        At each time $t$, we query mechanism $\blM^t$ from $\cAbl$, and construct a mechanism $\jaM^t$ by \emph{clipping} it to the square of interest $Q=[0,\nicefrac 12]\times[\nicefrac 12,1]$ (see first plot in \Cref{fig:reduction} in the main body), and applying the {inverse of} $f$ defined above. Formally, if $\blA^t$ is the allocation region of $\blM^t$, we define the allocation region $\jaA^t$ of $\jaM^t$ as follows: 
        \begin{equation}
        \label{eq:allocation_reduction}
            \jaA^t = \left\{(x',y')|f(x',y') \in \blA^t  \cap Q \right\}.
        \end{equation}
        Then, we observe the valuation $(v_1^t,v_2^t)$ generated from the joint-ads adversary, and feed $(v_s^t,v_b^t) = {f(v_1^t,v_2^t)}$ 
        back to the bilateral trade algorithm $\cAbl.$ The rest of the proof is devoted to proving that this is indeed a reduction with the desired properties. 
        
        Fix any $\sigma'$-smooth instance generated by the joint ads adversary. Since the bilateral trade algorithm sees the valuations generated by the joint-ads-$\sigma'$-smooth adversary mapped back through {$f$},
        then a simple calculation yields that $\cAbl$ is playing against a $\sigma$-smooth distribution, and thus is guaranteed the following regret bound by hypothesis:
        \begin{equation}
        \label{ineq:bl_regret}
            \sup_{\blM\in \cMbl}\sum_{t=1}^T \E{\prof_t(\blM) - \prof_t(\blM^t)} \le \rho(\sigma,T).
        \end{equation}
        Note, the $\prof_t(\cdot)$ functions compute the profit collected by a bilateral trade mechanism under valuations $(v_s^t,v_b^t)$ constructed as in the reduction.

        The next step is crucial: we need to relate the revenue $\rev_t(\jaM^t)$ extracted by the mechanism $\jaM^t$ in the joint ads instance with the profit $\prof_t(\blM^t)$ extracted by the mechanism $\blM^t$ in the corresponding bilateral trade instance.
        \begin{claim}
        \label{cl:reduction}
            The following inequality is verified: $\prof_t(\blM^t) \le \frac 12\rev_t(\jaM^t)$.
        \end{claim}
        \begin{proof}[Proof of \Cref{cl:reduction}]

        First, we observe that by construction $(v_s^t,v_b^t)$ belongs to $Q$. Therefore $(v_1^t,v_2^t)$ belongs to the allocation region of $\jaM^t$ if and only if the corresponding bilateral-trade-valuation $(v_s^t,v_b^t)$ belongs to the allocation region of $\blM^t$. We are then sure that when the trade does not happen in the bilateral trade instance, the same happens also in the joint ads one, and vice-versa. 
        
        Consider now the case in which both $(v_1^t,v_2^t) \in \jaA^t$ and the corresponding $(v_s^t,v_b^t)$ belongs to $\blA^t$ (and also to $Q$ by construction). We have the following:
        \begin{align}
        \notag
            \prof_t(\blM^t) &= \min\{y \, \mid (v_s^t, y) \in \blA^t  \} - \max\{x \, \mid (x, v_b^t) \in \blA^t  \}\\
        \notag 
            &\le \min\{y \, \mid (v_s^t, y) \in \blA^t \cap Q\} - \max\{x \, \mid (x, v_b^t) \in \blA^t \cap Q  \}\\
            &= \min\{y \, \mid f^{-1}(v_s^t, y) \in \jaA^t\} - \max\{x \, \mid f^{-1}(x, v_b^t) \in \jaA^t\} \tag{As $f$ is bijective}\\
        \notag 
            &= \min\{y \, \mid (v_1^t, 2y-1) \in \jaA^t\} - \max\{x \, \mid (1-2x, v_2^t) \in \jaA^t\}\\
        \notag
            &=\min\left\{\frac{1+y'}{2} \, \mid (v_1^t, y') \in \jaA^t\right\} - \max\left\{\frac{1-x'}2 \, \mid (x', v_2^t) \in \jaA^t\right\}\\
        \notag
            &=\frac 12 \min\left\{y' \, \mid (v_1^t, y') \in \jaA^t\right\} + \frac{1}2 \min\left\{x' \, \mid (x', v_2^t) \in \jaA^t\right\}\\
        \label{ineq:profit_rev}
            &= \frac 12 \rev_t(\jaM^t).
        \end{align}
        This concludes the proof.
        \end{proof}

        Consider the best mechanism $\jaM^\star$ for the $\sigma'$-smooth instance of joint ads\footnote{Also for joint ads one can prove that the $\sup$ in the regret definition is attained, but this proof carries over with minimal modifications by considering some small precision $\eps$ and the corresponding $\eps$-optimal mechanism.}, and denote with $\blM^\star$ the mechanism for bilateral trade constructed analogously to \Cref{eq:allocation_reduction}. We have the following result. 
        \begin{claim}
        \label{cl:benchmark_reduction}
            The following equality holds for any $t$: \(\prof_t(\blM^\star) = \frac 12\rev_t(\jaM^\star)\).
        \end{claim}
        \begin{proof}[Proof of \Cref{cl:benchmark_reduction}]
            The claim follows by exactly the same arguments of \Cref{cl:reduction}. The only difference is that, by construction, the allocation region of $\blM^\star$ is contained in $Q$, therefore the only inequality of \Cref{ineq:profit_rev} holds with equality.
        \end{proof}
        We have all the ingredient to bound the regret of $\cAja$: 
        \begin{align*}
            \sum_{t=1}^T&\E{\rev_t(\jaM^\star) - \rev_t(\jaM^t)} \\
            &\le 2\sum_{t=1}^T\E{\prof_t(\blM^\star) - \prof_t(\blM^t)} \tag{By Claims~\ref{cl:reduction} and \ref{cl:benchmark_reduction}}\\
            &\le {2\rho(\sigma,T),} 
            \tag{By \Cref{ineq:bl_regret}}
        \end{align*}
        which concludes the proof of the theorem.
        \end{proof}

        The reduction has two immediate corollaries. First, we can apply directly the reduction to translate the regret rate of \hier proved in \Cref{thm:hierHedge} to the joint ads problem. 
        \upperboundja*

        As a second corollary, we can apply the reduction in the other direction, and claim that no algorithm for bilateral trade can hope to achieve better than $\Omega(\sqrt{T})$ regret \emph{given} the lower bound of the same rate for joint ads \citep[see Theorem 8 of][]{AggarwalBDF24}.
        \lowerboundbl*

\end{document}